\documentclass{article}
\usepackage{graphicx}
\usepackage{latexsym}
\usepackage{color}

\usepackage{amsmath}
\usepackage{amssymb}
\usepackage{amsthm}
\usepackage{geometry}
\usepackage{graphicx}
\newtheorem{prop}{Proposition}

\allowdisplaybreaks
\begin{document}
\title{Perturbing Inputs to Prevent Model Stealing}

\author{Justin Grana\thanks{      \copyright 2020 IEEE. Personal use of this material is permitted. Permission from IEEE must be obtained for all other uses, in any current or future media, including reprinting/republishing this material for advertising or promotional purposes,creating new collective works, for resale or redistribution to servers or lists, or reuse of any copyrighted component of this work in other works.}  
  \thanks{I would like to thank Gavin Hartnett, Andrew Lohn,
    Adrienne Propp, A. Christian Johnson, Brian Vegetabile and other
    participants at the AI study circle for their feedback and
  comments.  The paper greatly improved because of them.}
  \\ Microsoft Corporation and Pardee RAND Graduate School\\
  justin.grana@microsoft.com}

  \maketitle

\begin{abstract}
We show how perturbing inputs to machine learning services
(ML-service) deployed in the cloud can protect against model stealing
attacks.  In our formulation, there is an ML-service that receives
inputs from users and returns the output of the model.  There is an
attacker that is interested in learning the parameters of the
ML-service. We use the linear and logistic regression models to
illustrate how strategically adding noise to the inputs fundamentally
alters the attacker's estimation problem.  We show that even with
infinite samples, the attacker would not be able to recover the true
model parameters.  We focus on characterizing the trade-off between
the error in the attacker's estimate of the parameters with the error
in the ML-service's output.
\end{abstract}

\section{Introduction}
In response to the advent of machine learning and artificial
intelligence, several firms have developed cloud-based machine-learning
as-a-service (ML-service) platforms \cite{mls}.  These platforms allow
users to access a trained machine learning model by submitting queries
and receiving output from the model.  The main value of such services
is that the end-users --- individuals that are typically only
interested in the model's output --- do not need to preoccupy
themselves with the model's details or the model training procedure.

As a simple example, an employee for a mortgage company might be
interested in the probability that a potential customer defaults on
their loan in the next year.  The employee can upload characteristics
of the customer such as age, educational attainment and outstanding
debt to the ML-service.  The pre-trained ML-service model will then
output back to the employee the probability that the potential
customer defaults on their mortgage.  In this example, the employee
receives the relevant information without having to understand the
details of the underlying prediction model.

Unfortunately, these models are subject to several types of attacks
\cite{memb,inv1,sum1}.  One particularly prevalent type of attack is
one in which an attacker conducts a model-stealing attack
\cite{ms1,ms2,ms3,knock,hf,bd}.  In such an attack, a malicious actor
attempts to learn the underlying model of the ML-service.  In this
paper, we present a method for protecting against model stealing
attacks where the attacker is directly interested in learning the
parameters of the model.  Our method relies on strategically adding
noise to the \emph{inputs} of the machine learning model which
introduces endogeneity into the attacker's estimation problem, and
thus an attacker's estimate of the model's parameters are biased and
inconsistent \cite{end1,end2,end3}.  Endogeneity arises in a
regression model when the error term is correlated with one or more of
the regressors.  In models with endogeneity, even with infinite
samples the attacker cannot learn the true parameters of the trained
model.  The main trade-off our method considers is how much noise in
the output of the ML-service is needed to induce a bias of a certain
size into the parameter estimates.  This trade-off is similar to other
methods aimed at preventing model stealing \cite{ms4}, however ours is
the first to consider the noise in the output as a function of noise
in the input.

There are several reasons why an attacker would have an incentive to
steal the ML-service model parameters.  First, ML-services usually charge per
query.  Therefore, if a user can learn the parameters of an ML-service
model, they can build their own prediction API and no longer have to
pay for queries.  Alternatively, if the data a user submits to the
platform is sensitive, then the user might wish to build its own model
to bypass any security concerns associated with sending queries over a
network.

Yet another reason ML-service parameters are subject to stealing is
because attackers have an inherent value in knowing the parameters.
For example, consider an ML-service platform that reports the
probability that an individual purchases a product conditional on
whether the individual saw a targeted internet advertisement.  A firm
that is considering placing internet advertisements might be
interested in stealing the model to learn the parameters that
determine the average treatment effect (ATE) of the advertisement.
Learning this parameter would inform the firm how much they would be
willing to pay for targeted advertising.



Our method protects against both an attacker that steals the
ML-service's parameters with the intention of replicating the
ML-service as well as an attacker that has an inherent value in
learning the parameters.  
As a point of clarification, our method does not directly address the
case in which an attacker wants to steal the ML-service to replicate
its prediction power but does so by other means than parameter
stealing.  For example, we don't explicitly address how our method
protects against an attacker that builds a neural network to replicate
the ML-service model to leverage its predictive power only with no
concern for the actual ML-service parameters.  Nevertheless, our
method still protects against a broad class of attackers.

It is also crucial to emphasize that this initial proof of concept is
simplified in two main ways.  First, we only consider linear and
logistic regression ML-services.  The properties of endogeneity do
indeed extend to more sophisticated models, including neural networks
\cite{nne}, but using such techniques would unnecessarily complicate
the analysis and detract from the key implications of adding noise to
ML-service inputs.  Secondly, we only consider relatively simple
attacker models.  A natural extension would be to include a full
game-theoretic formulation complete with attacker and defender utility
functions. While that is a valuable extension, this work is meant to
illustrate the efficacy of a defender whose strategy set contains
options on how to add noise to ML-service inputs.  

\section{Related Work}
The literature on model stealing is relatively young.  Initial
research showed that if an ML-service is a regression model and the
attacker knew the functional form of the model, it would be able to
recover the model parameters if it submitted one query per model
parameter \cite{ms1}. Other work included methods for stealing neural
networks and decision tree models from ML-services \cite{ms1,ms5} as
well as model hyper parameters \cite{ms2}.  Methods to protect against
model stealing are also a topic of interest, where most focus on
selectively perturbing the output of the ML-service \cite{ms4}.
Output perturbations include rounding the output or truncating the
output to obfuscate the ML-service's true output. \footnote{A small
  body of research focuses on detecting adversarial patterns in a
  user's sequence of queries \cite{ms3}}

Figure \ref{fig:dif} contrasts our approach with the literature.  In
the figure, the leftmost vertical sequence represents an ML-service
without any protection against model stealing.  In other words, inputs
are submitted to the service and the service returns an output.  The
top path is the traditional method of preventing model stealing by
perturbing the output.  This is the case when random noise is added to
the output \cite{ms4}, or instead of returning class probabilities, the
ML-service only returns a label \cite{lm}.  What such a transformation
effectively does is transforms the ML-service's ``True Model'' ($M$)
to another model, Model$^\prime$ ($M^\prime$).  However, perturbations
to the outputs do not guarantee that $M$'s parameters cannot be known
simple by sampling $M^\prime$.  For example, if the only difference
between $M$ and $M^\prime$ was that $M^\prime$ adds a mean zero
normal random variable to the output of $M$, an attacker would be able
to learn the parameters of $M$, just by querying $M^\prime$
sufficiently many times.

\begin{figure*}
  \centering
    \includegraphics[trim={0cm 0 0cm 0},clip, width=4in]{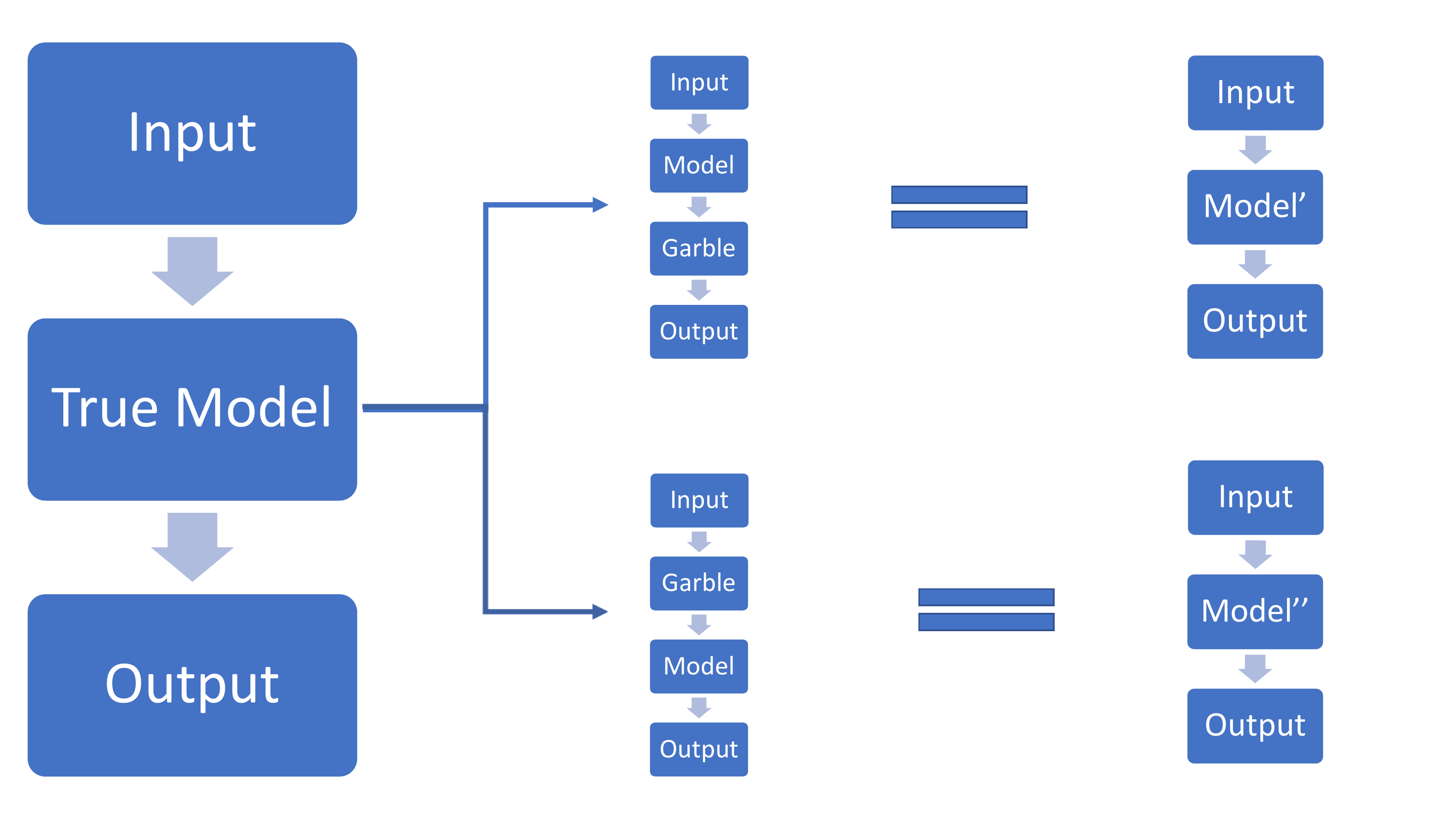}
    \caption{Two approaches to perturbing models}
    \label{fig:dif}
  \end{figure*}

Our approach is similar to the literature in that by perturbing the
inputs we are changing the ML-service's model from $M$ to
Model$^{\prime\prime}$ ($M^{\prime\prime}$), as shown in the bottom
path of figure \ref{fig:dif}.  However, our approach illuminates
several novelties that were obscured in the case of perturbed outputs.
First, we introduce a simple garbling function that makes it
impossible for the attacker to learn the parameters of $M$, even with
infinite samples of $M^{\prime\prime}$. This is an improvement over
the case when mean-zero noise is added to the ML-services output.
Secondly, we show that in non-linear models, it is possible to perturb
the inputs to obscure the parameters from the attacker while
maintaining a negligible difference in the output $M$ and
$M^{\prime\prime}$.  Finally, we show how to leverage the correlation
structure of the inputs to further reduce the difference between the
outputs of $M$ and $M^{\prime\prime}$ while still obscuring the true
parameters from the attacker.  Since previous models focus only on
perturbing outputs, this notion of leveraging the structure of inputs
to prevent model stealing has yet to be explored.  Of course, with
infinite samples the attacker would be able to approximate
$M^{\prime\prime}$.  However, we are interested in the case where the
defender wants to obscure the parameters of $M$, not
$M^{\prime\prime}$ and we examine how much noise the defender must
tolerate in order to obscure the parameters of $M$ by a given amount.


\section{General Framework}
A machine learning model is a function
$f_{\theta}:\mathcal{X} \rightarrow \mathcal{Y}$.  The function
$f_{\theta}$ is from $\mathbb{R}^K\rightarrow R^M$, and is governed by
parameters $\theta \in \Theta \in \mathbf{R}^P$.  If the ML-service did not need
to be protected from model stealing, a user would submit a vector
$\mathbf{x}$ and the ML-service would return the value
$f_{\theta}(\mathbf{x})$.  Previous work on preventing model stealing
focuses on changing the function $f_{\theta}$ so that when an attacker submits
a vector $\mathbf{x}$, the ML-service returns
$\tilde{f}_{\tilde{\theta}}(\mathbf{x})$.  Our approach is  different.
Instead of changing $f$ directly, we change $\mathbf{x}$ before it
goes into the machine learning service.  Specifically, there is a garbling
function $g:\mathcal{X}\rightarrow \mathcal{X}$ such that when a user
submits an input $\mathbf{x}$, the ML-service returns
$f_{\theta}(g(\mathbf{x}))$.

An attacker submits a dataset $\mathbf{X} = \{\mathbf{x}_1 ,
\mathbf{x}_2 \cdots \mathbf{x}_N\} \in \mathcal{X}^N$ to the machine
learning model and receives in return the predictions
$\mathbf{\hat{Y}} = \{f(g(\mathbf{x}_1)), f(g(\mathbf{x}_2)), \cdots
f(g(\mathbf{x}_N))\} \in \mathcal{Y}^N$.  Although not necessary, we
constrain $g$ such that for some dataset $\mathbf{X}$, the expected
deviation between $\mathbf{x}$ and $g(\mathbf{x})$ is $0$.
Specifically
\begin{align}
\frac{1}{|\mathbf{X}|}\sum_{\mathbf{x}\in \mathbf{X}}\mathbf{E}_{g}[\mathbf{x}- g(\mathbf{x})] = 0
\end{align}
Again, this is not a necessary constraint but we only include it to show
how our method is valid under mean-zero expected error, which might be
a desirable property.  Of course, loosening this constraint does not
degrade the performance of our method. In addition, our method does
not rely on the attacker submitting an entire dataset at once.
Specifically, our method is also valid when an attacker submits
sequential queries.  However, in that case the average error between
the true inputs and the garbled inputs would not, in general, be $0$.
Again, this does not degrade the performance of our method.

For all $N$, an attacker has an estimating function $h:
\mathcal{X}^N\times \mathcal{Y}^N \rightarrow \Theta$ that maps the
attacker's inputs and outputs to its estimate of the parameters,
$\hat{\theta}$.  Common estimating functions are the least squares or
maximum likelihood estimates.  To evaluate the effectiveness of our
method we consider two quantities:

\begin{enumerate}
  \item The difference between the attacker's estimate of the
    parameters and the true parameters with infinite samples, which is given by
    \begin{equation}
      \mathcal{D} =  \text{plim} h(\mathbf{X}, \mathbf{\hat{Y}}) - \theta 
    \end{equation}
          We call this the attacker's ``estimation error''.
  \item The expected squared difference between the ML-service's output without
    perturbing the input and the ML-service's output with the
    perturbed output, which is given by
    \begin{equation}
      \sigma^2 = \mathbf{E}\left[\left(f_{\theta}(\mathbf{x}) - f_{\theta}(g(\mathbf{x}))\right)^2\right]
    \end{equation}
    We call this the ML-service's ``prediction error.''  However, it
    is \emph{not} the difference between what the ML-service predicts
    and some true value but instead is the difference between what the
    ML-service predicts when it garbles the inputs and what the
    ML-service would predict if it didn't garble to inputs.  
\end{enumerate}

The equation for $\mathcal{D}$ is the difference between the true
parameters and the asymptotic limit of the attacker's estimate.  The
second term is the mean squared error between the ML-services output
and what the ML-service would output if it wasn't concerned with
protecting against model stealing.  Broadly speaking, the goal is to
make $\mathcal{D}$ as high as possible while minimizing $\sigma^2$.
The intuition is that $\mathcal{D}$ is the error in the attacker's
estimate of the true parameters, which is what is used to prevent
model stealing.  However, the ML-service still needs to be useful to
normal users so it can't distort the output so much that the
ML-service becomes useless.  The ML-service is most useful when
$\sigma^2$ is low.

One subtlety is that if the attacker \emph{knew} the value of
$\mathcal{D}$, it would be able to adjust its estimates to correct for
the error.  Therefore, higher values of $\mathcal{D}$ would not
necessarily prevent model stealing.  In a full game theoretic
formulation, the defender would likely randomize over values of
$\mathcal{D}$ so that the attacker wouldn't be able to know exactly
how to correct for the bias in its estimates.  Since in this proof of
concept, we do not consider the full game theoretic effects, we
motivate high values of $\mathcal{D}$ through its deterrent effect.
Specifically, if the attacker knew that the defender could obfuscate
the parameters with a high degree of error, then a high value of
$\mathcal{D}$ would disincentivize the attacker from stealing the
model.  In other words, the attacker might not know exactly the value
of $\mathcal{D}$ but knows that the defender could obfuscate the
estimates \emph{up to} a certain amount and thus, the attacker would
not be willing to steal the model because it knows that it's final
estimates might be prohibitively inaccurate.

\section{Results}
In this section we present a mix of analytical and numerical results.
We begin with the simplest case of a one covariate linear regression
model.  The linearity of the model permits an analytical analysis.  We
then extend the model to the case of logistic regression and examine
our method's performance We conclude with an analysis of correlated
regressors.

\subsection{Simple Linear Regression}
\label{sec:lin}
In the simplest case, consider the one variable linear machine
learning model $f_{\theta} = \alpha + \beta x$ where $\theta =
\{\alpha, \beta\}$.  Given a dataset $\mathbf{X}$ of queries, to
induce endogeneity we define the garbling function of a specific input
$x$ as
\begin{equation}
  g(x) = x + \gamma\left(\mu(x, \lambda^2)\right)
  \label{eq:garb}
\end{equation}
  where $\gamma$ and $\lambda$ are constants and $\mu(a,b)$ is a
  normal random variable with mean $a$ and variance $b$.
  Given the dataset $\mathbf{X}, \mathbf{\hat{Y}}$, the
  attacker's estimation function, $h$, is the ordinary least squares
  solution.  That is, the function $h$ for the attacker is ``choose the
  parameters $\hat{\alpha}$ and $\hat{\beta}$ that minimizes
  $\displaystyle\sum_{i=1}^N (\hat{y}_i - (\hat{\alpha}
  +\hat{\beta}x_i))^2$, where $\hat{y}$ is the value returned from the
  ML-service.  Since $\beta$ gives the marginal impact of a change in
  $x$, we focus the rest of the analysis on the slope parameters and
  ignore the intercept.

Before deriving the terms for $\mathcal{D}$ and $\sigma$, it is
important to emphasize that this ``simple'' formulation is not as
restrictive as it seems.  First, if instead of using least squares,
the attacker used a maximum likelihood estimation procedure, the
results would not change.  Secondly, the restriction of only one
regressor is also not as restrictive as it seems.  For example, if
there were multiple uncorrelated regressors and only one of the
regressors was perturbed, the results would not change.  Therefore,
this simple case actually covers multiple linear regression when one
regressor is garbled.  Later, we consider the case of multiple
perturbations of correlated regressors.  

We now derive the relationship between $\mathcal{D}$ and $\sigma^2$ in
the linear model.
\begin{prop}[Large Sample Results]
  \label{prop:linear}
  The relationship between the error in the attacker's estimate and
  the root mean square error is given by:
  \begin{enumerate}
    \item $\mathcal{D} = $\emph{plim}$(\hat{\beta}) - \beta  = \beta\gamma$
    \item $\sigma^2 = \left[(\beta\gamma)^2(Var(x)+\lambda^2)\right]$
  \end{enumerate}
\end{prop}
\begin{proof}
  First we prove 1.  Under the garbling function $g$, the data generating process for the
  ML-service is
  \begin{align}
  \hat{y} &= \alpha + \beta \left(x + \gamma\left(\mu(x, \lambda)\right)\right) \nonumber \\
  &= \alpha + \beta(x + \gamma(x+\lambda\epsilon)) \nonumber \\
  &= \alpha + (1+\gamma)\beta x + \beta\lambda\gamma\epsilon \label{eq:error} \\
  &= \alpha + \tilde{\beta}x + \eta \nonumber
\end{align}
where $\tilde{\beta} = (1+\gamma)\beta$, $\epsilon$ is a standard normal
random variable and $\eta$ is a composite mean-zero noise term.
This shows that the data generating process for the ML-service is a
simple linear model in $x$.  Therefore, when the attacker estimates
$\hat{\beta}$ from the data it estimates a linear regression model
with slope parameter equal to $\tilde{\beta}$.  Because OLS estimates
converge to the true parameters of the data generating process, the
attacker's estimate converges to $\tilde{\beta}$ and thus
$\tilde{\beta} - \beta = \beta\gamma$.

To prove part 2 of proposition \ref{prop:linear}, note that
\begin{align}
 \sigma^2 &= \mathbf{E}\left[\left(f_{\theta}(\mathbf{x})-f_{\theta}(g(\mathbf{x}))\right)^2\right]= \mathbf{E}[(y-\hat{y})^2] \nonumber \\
  &= \mathbf{E}\left[\bigg(\alpha + \beta x - (\alpha + \beta(1+\gamma)x + \beta\gamma\lambda\epsilon)\bigg)^2\right] \nonumber \\
  & = \mathbf{E}\left[\bigg((\beta\gamma)^2x^2 + 2\beta^2\gamma^2\lambda x\mu +(\gamma\beta\lambda\epsilon)^2\bigg)\right] \nonumber \\
  & = \left[(\beta\gamma)^2\mathbf{E}[x^2] + 2\beta^2\gamma^2\lambda\mathbf{E}[x\mu]+(\gamma\beta\lambda)^2\mathbf{E}[\epsilon^2]\right] \nonumber \\
  & = \left[(\beta\gamma)^2(Var(x) +\lambda^2)\right]
\end{align}
where the last line follows because $x$ and $\epsilon$ are independent and
$\epsilon$ is a standard normal variable.
\end{proof}

Proposition \ref{prop:linear} gives the precise trade-off between the
attacker's error in the parameter estimates and the noise in the
output required to induce the estimates.  The proposition reveals
several high level features.  First, the error in the attacker's
estimate is linear in $\beta$ and $\gamma$.  That means the attacker's
\emph{relative} error is linear in $\gamma$ only.  Second, the
prediction error depends on the variance of $x$.  This is an artifact
of having systematic errors.  Specifically, larger absolute values of
$x$ lead to higher errors because larger values of $x$ are perturbed
more. Although the linear relationship between the prediction error
and the variance of $x$ is a product of the functional form of the
garbling function, the main insight is that the prediction error will
depend on the distribution of $x$ because endogeneity requires
systematic errors.

Finally, the proposition shows that the prediction error is increasing
in the free parameter, $\lambda$.  Since the goal is to minimize
$\sigma^2$, an obvious choice is to set $\lambda$ to $0$.  However,
there are several exogenous reasons why the defender would not do
this.  First, $\lambda$ enters the equation for the attacker's
estimation error in finite samples as given in the following
proposition:

\begin{prop}[Small Sample Results]
  Suppose the attacker submits inputs of size $n$ and estimates
  $\beta$ using least squares or maximum likelihood.  Then:
  \begin{equation}
\mathbf{E}[(\beta-\hat{\beta})^2] = (\beta\gamma)^2 +\left(\beta\gamma\lambda\right)^2
\mathbf{E}\left[\left(\displaystyle\sum_{i=1}^n
  x_{i}^{2}\right)^{-1}\right]
  \end{equation}
  \label{prop:small}
\end{prop}
\begin{proof}
  The true data generating process with garbling is given by:
  \begin{equation}
    \hat{y}= \alpha + (1+\gamma)\beta x + \beta\lambda\gamma\mu
  \end{equation}
which is a standard linear regression model whose finite $n$
distribution is given in  \cite{end1} chapter $4$
\end{proof}
Consequently, if the defender is interested in not just the infinite
sample limit but also the small-sample parameter estimates, changing
$\lambda$ gives the defender another lever to obfuscate the model.
Secondly, the defender might also be interested in protecting the
ML-service from being replicated by an attacker (i.e. protecting an
attacker from stealing $M''$).  If $\lambda=0$, then the ML-service is
deterministic and it is relatively simple for an attacker to replicate
the ML-service.  Therefore, setting $\lambda>0$ protects against
stealing $M''$.  Finally, from a strategic perspective, a
sophisticated attacker would form beliefs about the garbling function
and the associated parameters ($\gamma$ and $\lambda$) in order to
best reconstruct an estimate for $\beta$.  If the defender always sets
$\lambda=0$, then the attacker only needs to reason about $\gamma$ to
make inference about the value of $\beta$, which reduces the
dimensionality of the attacker's decision problem.

Finally, if the attacker knew the nonzero value of $\lambda$, it
would be able to recover $\beta$, as given by the following
proposition.  
\begin{prop}
  \label{prop:needlam}
  Suppose the attacker knows the value of $\lambda$ and the attacker
  estimates the regression coefficients $\hat{\beta}$ and
  $\hat{\alpha}$ and the model variance $\hat{\Sigma^2} =
  \sum_{i=1}^N(\hat{y} - (\hat{\alpha} + \hat{\beta}x_i))^2$ using
  least squares or maximum likelihood. Then in
  the infinite limit, the attacker can identify $\beta$ as
  $\hat{\beta} - \frac{\hat{\Sigma}}{\lambda}$.
\end{prop}
\begin{proof}
As stated in proposition \ref{prop:linear}, $\hat{\beta}$ converges to
$(1 + \gamma)\beta$.  Since the attacker is estimating a linear model,
the attackers estimate of the error $\hat{\Sigma}^2$ converges to the
error in the true model which is given by $(\beta\gamma\lambda)^2$ in
equation \ref{eq:error}. This gives the following two equations where
the unknowns are $\beta$ and $\gamma$:
\begin{align}
  \hat{\Sigma}^2 = (\beta\gamma\lambda)^2 \nonumber \\
  \hat{\beta} = \beta(1+\gamma)
\end{align}
Solving for $\beta$ yields
\begin{equation}
  \beta = \hat{\beta} - \frac{\hat{\Sigma}}{\lambda}
\end{equation}
and thus $\beta$ is identified.  
\end{proof}

Proposition \ref{prop:needlam} illustrates how for our method to be
effective, a sophisticated attacker cannot know $\lambda$.  If it did,
it would be able to recover the true value of $\beta$.  Assuming
$\lambda>0$, this means in a full game theoretic treatment the
defender would likely randomize the values of $\gamma$ and $\lambda$
so that the attacker couldn't recover $\beta$ from the data. \footnote{Another
possibility would be to also add a mean-zero noise term to the output.
So for example, suppose the defender garbles with $g$ but also adds
$\mu(0, \lambda_2^2)$ to the output.  Then the equations for
$\hat{\Sigma}$ and $\hat{\beta}$ would contain the variables $\beta$,
$\gamma$, $\lambda$ and $\lambda_2$.  Therefore, as long as the
attacker doesn't know any \emph{2} of the variables, it would not be
able to recover $\beta$.}

\subsection{Logistic Regression}
Injecting endogeneity is not unique to the linear model and also
applies to nonlinear models.  The logistic regression model is a
special type of regression model that takes inputs and outputs
$\hat{y}\in[0,1]$ and is a common tool to predict the input's
``class.''  For example, logistic regression can be used to model the
probability of recidivism, the probability an individual files
bankruptcy, the probability a computer network is under attack or in
the multinomial case, the probability that an image contains each of
the digits 0-9. 

Under the same garbling function, $g$, as before, the logistic ML-service model
with perturbed inputs (assuming the $0$ mean for $x$) is given by:
\begin{align}
  \hat{y} =& \frac{e^{\alpha +\beta g(x)}}{1+e^{\alpha +\beta g(x)}} \nonumber \\
          =& \frac{e^{\alpha +\beta(x + \gamma(x + \lambda \mu))}}{1+e^{\alpha +\beta(x + \gamma(x + \lambda \mu))}}
\end{align}
Since the binomial logistic regression model can be
transformed into a linear model that predicts the ``log odds'', the
analysis is the same as in section \ref{sec:lin}.  However, when
translating the log-odds back to probabilities, the results are
slightly different as given in figure \ref{fig:log1}.
\begin{figure}
  \centering
  \includegraphics[width=\columnwidth]{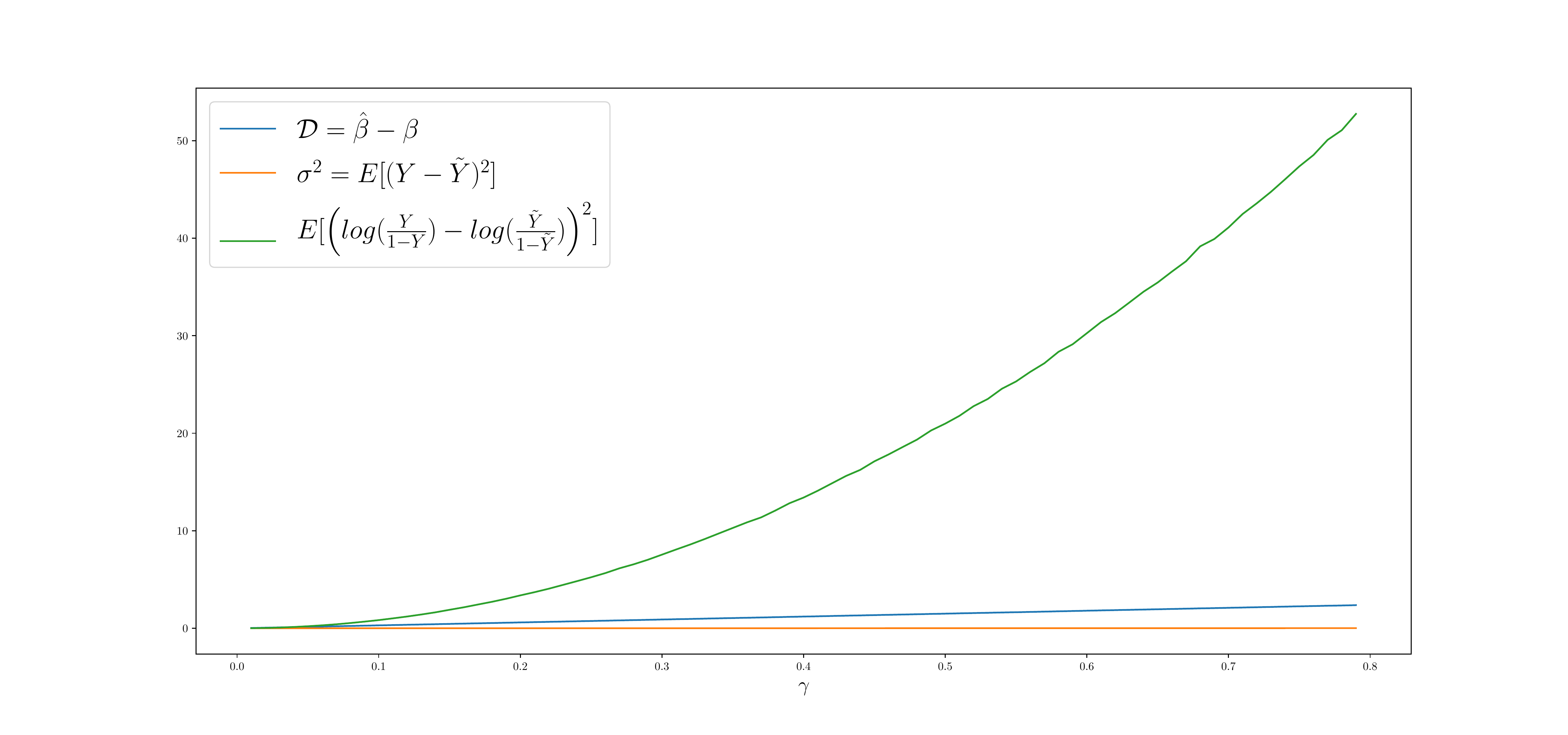}
  \caption{The figure shows that while the prediction error for the
    log odds is exponential in $\gamma$, the prediction error in the
    actual predicted probabilities is minimal.  This is due to the
    functional specification in the systematic errors in $X$.  For
    this plot, $\alpha=2, \beta=3, N=100000, \lambda = 1, Var(x)=8.3$}
  \label{fig:log1}
\end{figure}
The figure shows that it is possible to obscure the true value of the
parameters with a minimal effect on prediction error.  This is due to
the systematic perturbation of the inputs.  Specifically, the
perturbation of the inputs are most extreme for inputs that are high
in absolute value and mild for inputs low in absolute value.
Crucially, those larger perturbations occur along a relatively flat
portion of the logistic curve.  Therefore, the marginal effect of
perturbing the inputs at high absolute value inputs is minimal.  Of
course, the prediction error in the log-odds is still quadratic in
$\gamma$, as proposition \ref{prop:linear} established and figure
\ref{fig:log1} shows, but after transforming the log odds back to the
probabilities of interest, the prediction error is negligible.

Undoubtedly, this result is due to the functional form of $g$.  For
example, if $x$ values along the steepest part of the logistic curve
were perturbed the most, then different results would emerge.  While
determining what the optimal functional form of $g$ should be is
indeed an important and broad question, these results show that in
nonlinear models, certain choices of $g$ can allow the defender to
induce error in the parameter estimates with minimal impact on the
prediction error.  Specifically, perturbing inputs where the
derivative is close to zero provides minimal impact on the prediction
error.




\subsection{Multiple Perturbations}
Finally, we analyze the linear model with multiple regressors where each
regressor is garbled.  In this exposition, we consider the bivariate
case, though the results are easily generalized to $K$ regressors.
Specifically we consider the ML-service where $\hat{y}=f_{\theta} =
\alpha + \beta_1x_1 + \beta_2x_2$.  We now also perturb each $x_i$
independently so that
\begin{equation}
  g(x_i) = x_i + \gamma_i\left(\mu_i(x_i, \lambda)\right)
  \label{eq:garb2}
\end{equation}
where we assume the mean of $x_i$ is $0$. We could also of course let
$\lambda$ depend on $i$ but for ease of notation, we
assume $\lambda$ is the same for each regressor.  Like in the previous
section, we assume the attacker estimates the parameters using either
least squares of maximum likelihood.  We now establish the
multivariate version of proposition \ref{prop:linear}.

\begin{prop}
  \label{prop:multi}
  The relationship between the error in the attacker's estimate and
  the root mean square error is given by:
  \begin{enumerate}
    \item  \emph{plim}$(\hat{\beta_i}) - \beta_i  = \beta_i\gamma_i$
    \item $\begin{aligned}[t]
      &\sigma^2 = \bigg[(\beta_1\gamma_1)^2Var(x_1)+ (\beta_2\gamma_2)^2Var(x_2) \nonumber \\
        & +  2\beta_1\beta_2\gamma_1\gamma_2\emph{COV}(x_1,x_2)+ \\
        &(\lambda(\gamma_1\beta_1+\gamma_2\beta_2))^2\bigg] \nonumber
      \end{aligned}$
  \end{enumerate}
  where \emph{COV}$(x,y)$ is the covariance between random variables $x$ and $y$.  
\end{prop}
\begin{proof}
  As in the proof for proposition \ref{prop:linear}, it is
  straightforward to show that the ML-service's regression model with
  the perturbed inputs is given by
  \begin{align}
      \hat{y} = \alpha + (1+\gamma_1)\beta_1x_1 + (1+\gamma_2)\beta_2x_2
       \nonumber \\ + \lambda(\beta_1\gamma_1+\beta_2\gamma_2)\epsilon \;\;\;\;\;\;
  \end{align}
  where $\epsilon$ is a standard normal random variable.  By the same
  argument as in proposition \ref{prop:linear}, the attacker estimates
  of $\beta_i$ are multiplied by a factor of
  $(1+\gamma_i)$.  Proving $2$ can be derived as follows:
\begin{align}
 \sigma^2 &= \mathbf{E}\left[\left(f_{\theta}(\mathbf{x})-f_{\theta}(g(\mathbf{x}))\right)^2\right]= \mathbf{E}[(y-\hat{y})^2] \nonumber \\
 &= \mathbf{E}\bigg[\bigg(\alpha + \beta_1 x_1 + \beta_2x_2 \; -  \nonumber \\
   & \;\;\;\;\;\;\; \; \; \; \; \; \; \; \; \bigg(\alpha + \beta_1(1+\gamma_1)x_1 + \beta_2(1+\gamma_2)x_2 \; + \nonumber \\
   &\;\;\;\;\;\;\; \; \; \; \; \; \; \; \; \; \; \; \; \; \; \; \;\lambda(\beta_1\gamma_1+\beta_2\gamma_2)\epsilon\bigg)\bigg)^2\bigg] \nonumber \\
 & = \mathbf{E}\bigg[\bigg((\beta_1\gamma_1)^2x_1^2 + (\beta_2\gamma_2)^2x_2^2 + 2\beta_1\beta_2\gamma_1\gamma_2x_1x_2 \; + \nonumber \\
   & \;\;\;\;\;\;\;\;\;\;\;\;\;\;\;\; (\lambda(\beta_1\gamma_1+\beta_2\gamma_2)\epsilon)^2\bigg)\bigg] \nonumber \\
 & = \bigg[(\beta_1\gamma_1)^2\mathbf{E}[x_1^2] + (\beta_2\gamma_2)^2\mathbf{E}[x_2^2] + \nonumber \\
   & \;\;\;\;\;\;\;\; 2\beta_1\beta_2\gamma_1\gamma_2\mathbf{E}[x_1x_2]\;+ \nonumber \\
   & \;\;\;\;\;\;\;\; (\lambda(\beta_1\gamma_1+\beta_2\gamma_2))^2\mathbf{E}[\epsilon^2]\bigg] \nonumber \\
 & = \bigg[(\beta_1\gamma_1)^2Var(x_1)+ (\beta_2\gamma_2)^2Var(x_2) \; + \nonumber \\
   & \;\;\;\;\;\;\;\; 2\beta_1\beta_2\gamma_1\gamma_2\text{COV}(x_1,x_2)+ \nonumber \\
   & \;\;\;\;\;\;\;\;(\lambda(\beta_1\gamma_1+\beta_2\gamma_2))^2\bigg]
\end{align}
where the last line follows because $x_i$ and $\epsilon$ are independent and
$\epsilon$ is a standard normal variable.\footnote{It is
  straightforward to show that for the $K$ variate case, $\sigma^2 =
  \Gamma^T\Sigma\Gamma + \lambda^2||\Gamma||_2^2$ where $\Gamma$ is a $K\times
  1$ vector where the $i$'th element is $\beta_i\gamma_i$ and $\Sigma$
  is the covariance matrix of $X$.}

\end{proof}

The key difference between the simple univariate regression model and
the multiple regression case is that the ML-services prediction error
is influenced by the covariance of the regressors.  This leads to
several insights.  First, if the covariance is $0$ and the regressors
are uncorrelated, then the prediction error is the sum of the
prediction errors when only one of the regressors is perturbed.  More
importantly, for a non-zero covariance, the prediction error can be
less than the sum of the prediction error when only one variable is
perturbed.

To understand how the covariance of $x_1$ and $x_2$ impacts the
prediction error consider the following example.  Suppose that
$\beta_1,\beta_2, \gamma_1$ and $\gamma_2$ are positive and $x_1$ and
$x_2$ are negatively correlated so that the covariance is less than
$0$. Also, recall how the garbling function works.  In general, the
perturbations are such that, on average, positive values of $x_1$ and
$x_2$ are increased and negative values are decreased.  Furthermore,
due to the negative covariance, when $x_1$ is relatively high (and
positive), $x_2$ is relatively low (negative).  So, when a high value
of $x_1$ is perturbed, it is, on average, increased.  Since, by
assumption, $\beta_1$ is positive, the estimated output of $y$ would
be above the predicted value when there is no perturbation.  However,
at the same time $x_2$ is perturbed and, on average, decreased.  Since
by assumption, $\beta_2>0$, the estimated output of $y$ would be
\emph{below} the predicted value when there is no perturbation.  Or in
other words, the perturbation in $x_2$ cancels out some of the error
due to the perturbation in $x_1$.

The impact of the covariance is not limited to the case of negative
covariances.  Since $\gamma_1$ and $\gamma_2$ are free parameters, a
defender that sets $\gamma_1$ and $\gamma_2$ to ensure a certain
relative prediction error can choose the sign of $\gamma_1$ and
$\gamma_2$ so that the relative error in the parameter estimate is
unchanged but $\gamma_1\gamma_2\beta_1\beta_2\emph{COV}(x_1,x_2)<0$.
This ensures that the defender leverages the reduction in prediction
error due to correlated inputs.

\section{Discussion and Future Work}
We used regression models to show that it is possible to strategically
inject error into an ML-service's inputs to guard against an attacker
learning the model parameters.  Importantly, our scenario is different
from other scenarios in the literature as in our case, we assume the
attacker is interested in recovering the ML-service's parameters and
not only replicating the output of the ML-service.

For the univariate linear regression model, we introduced a simple
garbling function that transformed the inputs before being sent to the
ML-service.  We illustrated the trade-off between the attacker's
error in the estimation of the model's parameters and the error in the
machine learning model's output against what its output would be if it
didn't perturb the inputs.  Under our garbling function, the
(squared) prediction error scales quadratically in the error of the
estimate.

We showed that our results extend to the case of logistic regression.
However, while the scaling of the prediction and estimate error are
the same when considering the log-odds ratios, the relationship does
not hold when considering the error in the outputted probability.
Specifically, since our garbling function injects most of the noise in
the tails of the logistic curve, the garbling has little effect on the
prediction error of the ML-service.  While this is a simple example
and will likely not replicate for all types of ML-services, it shows that
carefully chosen garbling functions in non-linear models may perform
better than the same garbling function in a linear model.  Finally, we
showed how a defender can leverage correlation among the inputs to
reduce the ML-service's prediction error while not changing the
attacker's relative estimation error.

Our work, which was mainly analytical in nature, is a simple proof of
concept 
and to remain tractable, our results were
limited to simple functional form assumptions.  Real ML-services are
likely to be highly nonlinear (random forests and neural networks are
two primary examples).  An obvious extension to this proof-of-concept
is to investigate how injecting noise impacts prediction and
estimation in these large and complex models.  While analytical
results may be intractable, computational experiments may reveal how
features of the garbling function, for example, impact the trade-off
between prediction error and estimation error.  This would include an
example using a real data set. 

We also assumed relatively simple attacker models.  Therefore, an
obvious extension to this work is a full game theoretic treatment of a
strategic attacker and defender.  In such a game, the defender would
choose the garbling function and the attacker would form rational
beliefs about the garbling function and choose its estimation strategy
optimally, given its beliefs. This extension would both ground our
method in a rigorous decision theoretic framework and also illuminate
``worse case'' scenarios by assuming an infinitely rational attacker.

\bibliographystyle{IEEEtrans}
\bibliography{refs.bib}

\begin{thebibliography}{10}

\bibitem{end1}
W.H. Greene.
\newblock {\em Econometric Analysis}.
\newblock Pearson Education, 2003.

\bibitem{end3}
Zvi Griliches and Jerry~A Hausman.
\newblock Errors in variables in panel data.
\newblock {\em Journal of econometrics}, 31(1):93--118, 1986.

\bibitem{nne}
Jason Hartford, Greg Lewis, Kevin Leyton-Brown, and Matt Taddy.
\newblock Counterfactual prediction with deep instrumental variables networks.
\newblock {\em arXiv preprint arXiv:1612.09596}, 2016.

\bibitem{hf}
Matthew Jagielski, Nicholas Carlini, David Berthelot, Alex Kurakin, and Nicolas
  Papernot.
\newblock High-fidelity extraction of neural network models.
\newblock {\em arXiv preprint arXiv:1909.01838}, 2019.

\bibitem{ms3}
Mika Juuti, Sebastian Szyller, Samuel Marchal, and N~Asokan.
\newblock Prada: protecting against dnn model stealing attacks.
\newblock {\em arXiv preprint arXiv:1805.02628}, 2018.

\bibitem{ms4}
Taesung Lee, Benjamin Edwards, Ian Molloy, and Dong Su.
\newblock Defending against machine learning model stealing attacks using
  deceptive perturbations.
\newblock {\em arXiv preprint arXiv:1806.00054}, 2018.

\bibitem{lm}
Daniel Lowd and Christopher Meek.
\newblock Adversarial learning.
\newblock In {\em Proceedings of the eleventh ACM SIGKDD international
  conference on Knowledge discovery in data mining}, pages 641--647. ACM, 2005.

\bibitem{knock}
Tribhuvanesh Orekondy, Bernt Schiele, and Mario Fritz.
\newblock Knockoff nets: Stealing functionality of black-box models.
\newblock In {\em Proceedings of the IEEE Conference on Computer Vision and
  Pattern Recognition}, pages 4954--4963, 2019.

\bibitem{sum1}
Nicolas Papernot, Patrick McDaniel, Ian Goodfellow, Somesh Jha, Z~Berkay Celik,
  and Ananthram Swami.
\newblock Practical black-box attacks against machine learning.
\newblock In {\em Proceedings of the 2017 ACM on Asia conference on computer
  and communications security}, pages 506--519. ACM, 2017.

\bibitem{mls}
Mauro Ribeiro, Katarina Grolinger, and Miriam~AM Capretz.
\newblock Mlaas: Machine learning as a service.
\newblock In {\em 2015 IEEE 14th International Conference on Machine Learning
  and Applications (ICMLA)}, pages 896--902. IEEE, 2015.

\bibitem{ms5}
Yi~Shi, Yalin Sagduyu, and Alexander Grushin.
\newblock How to steal a machine learning classifier with deep learning.
\newblock In {\em 2017 IEEE International Symposium on Technologies for
  Homeland Security (HST)}, pages 1--5. IEEE, 2017.

\bibitem{memb}
Reza Shokri, Marco Stronati, Congzheng Song, and Vitaly Shmatikov.
\newblock Membership inference attacks against machine learning models.
\newblock In {\em 2017 IEEE Symposium on Security and Privacy (SP)}, pages
  3--18. IEEE, 2017.

\bibitem{ms1}
Florian Tram{\`e}r, Fan Zhang, Ari Juels, Michael~K Reiter, and Thomas
  Ristenpart.
\newblock Stealing machine learning models via prediction apis.
\newblock In {\em 25th $\{$USENIX$\}$ Security Symposium ($\{$USENIX$\}$
  Security 16)}, pages 601--618, 2016.

\bibitem{ms2}
Binghui Wang and Neil~Zhenqiang Gong.
\newblock Stealing hyperparameters in machine learning.
\newblock In {\em 2018 IEEE Symposium on Security and Privacy (SP)}, pages
  36--52. IEEE, 2018.

\bibitem{end2}
Tom Wansbeek.
\newblock Gmm estimation in panel data models with measurement error.
\newblock {\em Journal of Econometrics}, 104(2):259--268, 2001.

\bibitem{inv1}
Xi~Wu, Matthew Fredrikson, Somesh Jha, and Jeffrey~F Naughton.
\newblock A methodology for formalizing model-inversion attacks.
\newblock In {\em 2016 IEEE 29th Computer Security Foundations Symposium
  (CSF)}, pages 355--370. IEEE, 2016.

\bibitem{bd}
Huadi Zheng, Qingqing Ye, Haibo Hu, Chengfang Fang, and Jie Shi.
\newblock Bdpl: A boundary differentially private layer against machine
  learning model extraction attacks.
\newblock In {\em European Symposium on Research in Computer Security}, pages
  66--83. Springer, 2019.

\end{thebibliography}
\end{document}